\documentclass[11pt]{amsart}
\usepackage{amssymb}
\usepackage{comment}
\usepackage[shortlabels]{enumitem}
\setlist[enumerate,1]{leftmargin=20pt}
\setlist[description,1]{leftmargin=15pt}
\usepackage{hyperref}
\usepackage[leftmargin=15pt,vskip=8pt]{quoting}
\usepackage{url}
\usepackage[usenames]{color}

\newtheorem{theorem}{Theorem}[section]
\newtheorem{corollary}[theorem]{Corollary}
\newtheorem{lemma}[theorem]{Lemma}

\theoremstyle{definition}
\newtheorem{convention}[theorem]{Convention}
\newtheorem{definition}[theorem]{Definition}
\newtheorem{example}[theorem]{Example}
\newtheorem{observation}[theorem]{Observation}

\newtheorem{proviso}[theorem]{Proviso}
\newtheorem{remark}[theorem]{Remark}
\newtheorem{task}[theorem]{Task}

\newcommand{\Active}{\textrm{Active}}
\newcommand{\Done}{\textrm{Done}}
\newcommand{\Else}{\texttt{else }}
\newcommand{\Elseif}{\texttt{elseif }}
\newcommand{\false}{\texttt{false}}

\newcommand{\If}{\texttt{if }}
\newcommand{\ITE}{\texttt{ITE$\,$}}
\newcommand{\IN}{\textrm{Input}}
\newcommand{\initialized}{\textrm{initialized}}
\newcommand{\M}{{\mathbf M}}
\newcommand{\Max}{\texttt{Max}}
\newcommand{\n}{\texttt{n}}
\newcommand{\nil}{\texttt{nil}}
\newcommand{\OUT}{\textrm{Output}}
\newcommand{\pr}{\hspace{.1em}\parallel\hspace{.1em}}
\newcommand{\prl}{\hspace{.5em}\parallel\hspace{.5em}}
\newcommand{\Ret}{\texttt{Ret}}
\renewcommand{\t}{\bar t}
\newcommand{\Then}{\texttt{then }}
\newcommand{\To}{\texttt{To}}
\newcommand{\TOP}{\texttt{Top}}
\newcommand{\true}{\texttt{true}}
\newcommand{\U}{\Upsilon}
\newcommand{\V}{\mathcal V\hspace{1pt}}
\newcommand{\x}{\bar x}
\newcommand\xqed[1]{%
  \leavevmode\unskip\penalty9999 \hbox{}\nobreak\hfill\quad\hbox{#1}}
\newcommand\textqed{\xqed{$\triangleleft$}}

\title{Means-fit effectivity}
\author[Yuri Gurevich]{Yuri Gurevich\\
{\scriptsize University of Michigan}}
\address{Computer Science and Engineering\\
University of Michigan\\
Ann Arbor, MI  48109, U.S.A}
\email{gurevich@umich.edu}
\begin{document}
\thispagestyle{empty}

\begin{abstract}
Historically, the notion of effective algorithm is closely related to the Church-Turing thesis.
But effectivity imposes no restriction on computation time or any other resource; in that sense, it is incompatible with engineering or physics.
We propose a natural generalization of it, \emph{means-fit effectivity}, which is effectivity relative to the (physical or abstract) underlying machinery of the algorithm.
This machinery varies from one class of algorithms to another. Think for example of ruler-and-compass algorithms, arithmetical algorithms, and Blum-Shub-Smale algorithms.
We believe that means-fit effectivity is meaningful and useful independently of the Church-Turing thesis.
Means-fit effectivity is definable, at least in the theory of abstract state machines (ASMs).
The definition elucidates original effectivity as well.
Familiarity with the ASM theory is not assumed. We tried to make the paper self-contained.
\end{abstract}
\maketitle

\section{Introduction}
\label{sec:intro}

How can you prove the Church-Turing thesis? Here is one idea from our favorite logic text:

\begin{quoting}\small
``We get more evidence [for Church's thesis] if we try to define \emph{calculable} directly. For simplicity, consider a unary calculable function $F$. It is reasonable to suppose that the calculation consists of writing expressions on a sheet of paper (or that it can be reduced to this). As will become clear in the next section, there is no loss of
generality in supposing that the expressions written are numbers (more precisely, expressions which designate numbers). We therefore write $a_0, a_1, \dots, a_n$, where $a_0$ is $a$ and $a_n$ is $F(a)$. Now the decision method tells us how to derive $a_i$ from $a_0,\dots,a_{i-1}$ or, equivalently, from $\langle a_0, \dots, a_{i-1}\rangle$. Hence there is a calculable function $G$ such that $G(\langle a_0, \dots, a_{i-1}\rangle) = a_i$. The decision method also tells us when the computation is complete; so there is a calculable predicate $P$ such that $P(\langle a_0, \dots, a_i\rangle)$ is false for $i < n$ and true for $i = n$.

Our attempt to define calculability thus ends in circularity, since $G$ and $P$ must be assumed to be calculable. However, since $G$ describes a single step in the calculation, it must be a very simple calculable function; and the same applies
to $P$. We can therefore expect, on the basis of other evidence for Church's thesis, that $G$ and $P$ will be recursive. If we assume this, we can prove that $F$ is recursive (Shoenfield, \cite[\S6.5]{Shoenfield}). \textqed
\end{quoting}

\noindent
Shoenfield's idea%
\footnote{Actually we don't know whose idea it is; Shoenfield doesn't quote sources in the textbook.}
has been realized. We explain this below. But first let's recall the thesis: Every effective numerical function is (partial) recursive or, equivalently, Turing computable.
Here \emph{numerical functions} are partial functions $y = f(x_1,\dots,x_r)$ where the arguments $x_i$ range over natural numbers and the values $y$, if defined, are natural numbers.
A numerical function is effective if it is computable by an effective algorithm.

The notion of effectivity is famously elusive. In this paper, we propose a more general notion, \emph{means-fit effectivity}. It seems to us useful independently of the Church-Turing thesis, and it elucidates the original notion of effectivity as well.

But let us explain all this in an orderly fashion, starting with an important reservation:
This study is restricted to sequential algorithms, a.k.a.\ classical or traditional, the algorithms of the historical Church-Turing thesis.
A sequential algorithm is deterministic.
Its computations are finite or infinite sequences of steps.
And the computation steps are of bounded complexity. (This last property, observed by Kolmogorov in \cite{Kolmogorov}, rules out massive parallelism.)

Thesis-related literature, including Shoenfield's 1967 logic text \cite{Shoenfield}, was virtually restricted to sequential algorithms during the first few decades after the formulation of the thesis. One may be tempted to analyze all algorithms, but this is virtually impossible because the general notion of algorithm has not matured; it is evolving and the evolution may never stop \cite{G209}.

\begin{proviso}
By default, algorithms are sequential in the rest of this article. \textqed
\end{proviso}

This study is enabled by the axiomatization of sequential algorithms in \cite{G141}, specifically by the representation theorem in \cite[\S6]{G141}, according to which every sequential algorithm is behaviorally identical to a (sequential) abstract state machine (an ASM for short).
Since the only aspect of algorithms that we are interested in is their behavior, we work with ASMs and call them algorithms%
\footnote{The axiomatization of \cite{G141} is somewhat refined in \cite{G201}, but the representation theorem remains valid. Much of the analysis in \cite{G141} and in this study can be generalized to other species of algorithms which have been axiomatized, e.g., to synchronous parallel algorithms \cite{G157} and interactive algorithms \cite{G182}.}.

It is the representation theorem of \cite{G141} that realizes Shoenfield's idea, but the realization does not solve the problem of characterizing effectivity. Why not? Well, Shoenfield assumed that initially we have only input. In fact, we have also some basic operations available \emph{and nothing else}. This ``nothing else" is in essence the missing ingredient.
Adding it to the axiomatization of \cite{G141} allows one to derive the Church-Turing thesis in its core setting \cite{G188}; we touch on this issue in \S\ref{discuss}.

Why did the axiomatization in \cite{G141} allow ineffective algorithms? One reason is that we were interested primarily in engineering applications. While tidy initial states are natural in theoretical study, there may be nothing tidy about the initial states of some engineering algorithms. Those initial states might have been prepared --- and messed up --- by various processes.
Besides, abstracting from resource usage, inherent in the notion of effectivity, is incompatible with engineering.
But there is something else which is important in applications and is also relevant to our current story.

The realm of not-necessarily-effective algorithms is more natural. You use whatever tools are --- in reality or by convenient abstraction --- available to you. Turing's idealized human agent uses pen and paper: ``Computing is normally done by writing certain symbols on paper" \cite[\S9]{Turing}. Ruler and compass don't fit this description but they had been used in antiquity. In real-time engineering applications and in some theoretical computation models, like BSS \cite{BSS}, you work with genuine reals. Hence the interest in tool dependent, or means dependent effectivity.

To formalize the notion of means-fit effectivity, we use the ASM theory. The states of an abstract state machine are first-order structures with static and dynamic basic functions, where the dynamic functions play the role which is normally played by program variables in programming languages. Without loss of generality (as we show in \S\ref{sec:sd}), the static functions and input variables carry all the initial information.

The key idea of this study is a classification, in \S\ref{sec:f}, of static functions into intrinsic and extrinsic.
As far as an algorithm is concerned, all its static functions look like oracles, but the intrinsic functions are built-in functions provided reliably by the underlying machinery of the algorithm or, more abstractly, by the (in general compound) datastructure of the algorithm.
All intrinsic functions are effective relative to the datastructure.
On the other hand, extrinsic functions are provided by  outside entities with no guarantee of reliability in general.
Of course some or all extrinsic functions may be effective as well but such effective functions can be pruned off, as we prove in \S\ref{sec:prune}. (Pruning Theorem~\ref{thm:prune} is our main technical result.)
Accordingly, we define that an algorithm is \emph{means-fit effective} or \emph{effective relative to its datastructure} if it has no extrinsic functions.

What relevance does this have to the Church-Turing thesis? Well, let's consider how the thesis could possibly be falsified. One way is mathematical.
Construct an effective numerical function which is not recursive, as Ackermann (and independently Sudan) constructed  a recursive numerical function which is not primitive
recursive.
Another way may be called physical. Find new means, possibly using new discoveries of physics, which would allow you to effectively compute a numerical function that is not recursive.
To us, the physical version makes little sense. The abstraction from resources is incompatible with physics. In any case, predicting the future of physics is beyond the scope of this paper.

The historical Church-Turing thesis was mathematical, and it is the historical thesis that we discuss here and in \cite{G188}.
One particular datastructure (using our current terminology) was in the center of attention historically.  It was the arithmetic of natural numbers; see for example the quotation from Shoenfield's textbook above. Call algorithms with that datastructure \emph{arithmetical}. We believe that the historical thesis can be formulated thus: if a numerical function is computed by an effective numerical algorithm then it is partial recursive. This form of the thesis has been derived in \cite{G188} from the axioms of \cite{G141} plus an initial-state axiom according to which, initially, the dynamic functions of the algorithm --- with the exception of input variables --- are uninformative. The current paper provides additional justification for this initial-state axiom.
We return to this discussion in \S\ref{discuss} where we discuss other related work as well.

\subsection*{Acknowledgments}

I am extremely grateful to Andreas Blass who provided useful comments on all aspects of this paper, from high-level ideas to low-level details of exposition including definite/indefinite articles absent in my native Russian. (If you know all rules about the articles, explain this: the flu, a cold, influenza.)
I am also very grateful to Udi Boker, Patrick Cegielski, Julien Cervelle, Nachum Dershowitz, Serge Grigorieff and Wolfgang Reisig for most useful comments on short notice.

\section{Abstract state machines}
\label{sec:asm}

To make this paper self-contained and introduce terminology, we recall some basic notions of first-order logic in the form appropriate to our purposes.

\subsection{First-order structures}
\label{sub:fo}

A \emph{vocabulary} is a finite collection of function symbols, each of a fixed arity.
Some function symbols may be marked as \emph{relational} and called \emph{relations}.
Some function symbols may be marked as \emph{static}; the other function names are \emph{dynamic}.

Two vocabularies are \emph{consistent} if any symbol that belongs to both vocabularies has the same arity and markings in both vocabularies.

\begin{convention}[Vocabularies]\mbox{}
\begin{enumerate}
\item Every vocabulary contains the following {\em logic symbols:} the equality sign, nullary symbols $\true$,  $\false$, and $\nil$, the standard Boolean connectives, and a ternary function symbol \ITE\ (read if-then-else)\footnotemark.
\item All logic symbols, except $\nil$ and $\ITE$, are relational.
\item All logic symbols are static.  Nullary dynamic symbols are \emph{elementary variables}. Relational elementary variables are \emph{Boolean variables}. \textqed
\end{enumerate}
\end{convention}

\footnotetext{\ITE\ is a novelty introduced here, though we intended to do that already for a while. It is used below in \S\ref{sec:sd}.}

A \emph{(first-order) structure} $X$ of vocabulary $\U$ is a nonempty set $|X|$, the \emph{base set} of $X$, together with \emph{basic functions} $f_X$ where $f$ ranges over $\U$. (The subscript $X$ will be often omitted.) If $f$ is $r$-ary then $f_X$ is a function, possibly partial, from $|X|^r \to |X|$.

\begin{convention}[Structures]\mbox{}
\begin{enumerate}
\item The \emph{nonlogic vocabulary} of a structure $X$ is the vocabulary of $X$ minus all the logic symbols.
\item By default, basic functions are total. This guideline will save us space. Instead of indicating totality in most cases, we will indicate partiality in just a few cases.
\item In every structure, {\tt true} and {\tt false} and $\nil$ are defined and distinct.
\item Every (defined) value of every basic relation is either $\true$ or $\false$. The equality sign and the standard Boolean operations have their usual meaning.
\item $\ITE(x,y,z)$ is $y$ if $x=\true$, is $z$ if $x=\false$, and is $\nil$ otherwise. \textqed
\end{enumerate}
\end{convention}

\begin{remark}
In constructive mathematics, (constructive) real numbers are represented by algorithms. As a result, the equality of (thus represented) real numbers becomes partial. We took a similar position in \cite{G201}. But one may want to distinguish between genuine reals and their representations and to keep the equality of genuine reals total.
\end{remark}

\noindent
Here $\nil$ is an error value of sorts; the first argument of $\ITE$ is normally Boolean. $\nil$ replaces {\tt undef} of \cite{G141} to emphasize the difference between default/error values and the absence of any value.

Terms of vocabulary $\U$ are defined by induction: If $f$ is an $r$-ary symbol and $t_1,\dots,t_r$ are terms then $f(t_1,\dots,t_r)$ is a term. (Here the case $r=0$ is the basis of induction.) A term $f(\t)$ is \emph{Boolean} if $f$ is relational.

The value of an $\U$ term in a $\U$ structure $X$ is defined by induction:
\[\V_X f(t_1,\dots,t_r) =
 f_X\big(\V_X t_1, \dots, \V_X t_r\big).\]
By default, to evaluate $f(t_1,\dots,t_r)$, evaluate the terms $t_i$ first. But $\ITE(t_1,t_2,t_3)$ is an exception: Evaluate $t_1$ first. f the result is $\true$ then evaluate $t_2$ but not $t_3$, and if the result is $\false$ then evaluate $t_3$ but not $t_2$; otherwise evaluate neither $t_2$ nor $t_3$.

If $X_1, X_2$ are structures of vocabularies $\U_1, \U_2$ respectively, and $t_1, t_2$ are terms of vocabularies $\U_1, \U_2$ respectively, then $\V_{X_1}(t_1) = \V_{X_2}(t_2)$ means that either both sides are defined and have the same value (so that the two structures have common elements) or else neither side is defined.

The rest of this subsection is devoted to introduction of the union of structures. Call two structures are \emph{consistent} if their vocabularies are consistent and the following condition holds for every joint function symbol $f$. If $r$ is the arity of $f$ and elements $x_1,\dots,x_r$ belong to both structures then
$f(x_1,\dots,x_r)$ is the same in both structures.

\begin{definition}[Union]\label{def:union}\mbox{}
Let $X_1,\dots, X_N$ be pairwise consistent structures with the same logic elements.
The \emph{union} $X_1\cup X_1 \cup \dots \cup X_N$ of the structures $X_i$ is the structure $X$ such that
\begin{enumerate}
\item the vocabulary of $X$ is the  union of the vocabularies of $X_1,\dots, X_N$,
\item the base set of $X$ is the union of the base sets  of $X_1,\dots, X_N$, and
\item if $f$ is an $r$-ary basic function of $X_i$ then $\V_X(f(x_1,\dots,x_r))$ is the default value unless all $r$ elements $x_1,\dots,x_r$ belong to $X_i$\,, in which case $\V_X(f(x_1,\dots,x_r)) = \V_{X_i}(f(x_1,\dots,x_r))$ . \textqed
\end{enumerate}
\end{definition}

Notice that different structures $X_i$ may share nonlogic elements, in which case the union is not a disjoint union modulo the logic.

\subsection{Abstract state machines: Definition}
\label{sub:asmdef}

Traditionally, in logic, structures are static, but we will use structures as states of algorithms.
Let $X$ be a structure of vocabulary $\U$.
A {\em location} $\ell$ of $X$ is a pair $(f,\x)$ where $f\in\U$, $f$ is dynamic, $\x\in|X|^r$ and $r$ is the arity of $f$.  The {\em content} of location $\ell$ is $f_X(\x)$.
An {\em (atomic) update} of $X$ is a pair $(\ell,y)$ where $\ell$ is a location $(f,\x)$ and $y\in|X|$.  To execute the update $(\ell,y)$ means to replace the current content of $\ell$ with $y$, that is to set $f_X(\x)$ to $y$. This produces a new structure.
Updates $(\ell_1,y_1)$ and $(\ell_2,y_2)$ are \emph{contradictory} if $\ell_1=\ell_2$ but $y_1\neq y_2$; otherwise the updates are \emph{consistent}.

\begin{definition}[Rules]
\emph{Rules} of vocabulary $\U$ are defined by induction.
\begin{enumerate}
\item An \emph{assignment} has the form\quad
    $f(t_1,\dots,t_r):=t_0$\quad where\\
    $f\in\U$, $f$ is dynamic, $r$ is the arity of $f$, and $t_0,\dots,t_r$ are $\U$ terms.
\item A \emph{conditional rule} has the form\quad
    $\texttt{if $\beta$ then $R_1$ else $R_2$}$\\
    where $\beta$ is a Boolean $\U$ term and $R_1, R_2$ are $\U$ rules.
\item A \emph{parallel rule} has the form\quad
    $R_1 \pr R_2$\quad where $R_1, R_2$ are $\U$ rules. \textqed
\end{enumerate}
\end{definition}

\smallskip\noindent
\textbf{Semantics of rules.} A successful  execution of an $\U$ rule $R$ at an $\U$ structure $X$ produces a pairwise consistent finite set $\{(\ell_1,y_1), \dots, (\ell_n,y_n)\}$ of updates and thus results in a new state $X'$, obtained from $X$ by executing these updates.
An assignment $f(t_1,\dots,t_r):=t_0$ produces a single update $\big(\ell,\V_X(t_0)\big)$ where $\ell = \Big(f,\big(\V_X(t_1),\dots,\V_X(t_r)\big)\Big)$.
A conditional rule $\If \beta\ \Then R_1\ \Else R_2$ produces the update set of $R_1$ if $\beta=\true$ in $X$ and the update set of $R_2$ if $\beta=\false$ .
A parallel rule $R_1\pr R_2$ produces the union of the update sets of $R_1$ and $R_2$.

Notice that the execution of the assignment in a given state $X$ does not require the evaluation of the term $f(t_1,\dots,t_r)$. Let $x_i = \V_X(t_i)$ for $i = 0, \dots, r$. If $f(x_1,\dots,x_r)$ is undefined but $x_0$ is defined than $f(x_1,\dots,x_r) = x_0$ after the assignment.

\begin{definition}
A (sequential) ASM $A$ of vocabulary $\U$ is given by a  \emph{program} and \emph{initial states}. The program is a rule of vocabulary $\U$. Initial states are $\U$ structures. The collection of initial states is nonempty and closed under isomorphisms. \textqed
\end{definition}

A \emph{computation} of $A$ is a finite or infinite sequence $X_0, X_1, X_2, \dots$ of $\U$ structures where $X_0$ is an initial state of $A$ and where every $X_{i+1}$ is obtained by executing $\Pi$ at $X_i$. A \emph{(reachable) state} of $A$ is an $\U$ structure that occurs in some computation of $A$.

As we mentioned in \S\ref{sec:intro}, every (sequential) algorithm $A$ is behaviorally identical to some (sequential) ASM $B$; they have the same initial states and the same state-transition function.

\begin{proviso}
By default, algorithms are abstract state machines in the rest of this article.
\end{proviso}

For future use, we formulate the following obvious observation.

\begin{observation}[Failure]\label{obs:fail}
There are two scenarios that an algorithm $A$ fails at a given state $X$. One is that the algorithm attempts to evaluate a basic partial function $f$ at an input where $f$ is undefined. The other reason is that the program of $A$ produces contradictory updates of some basic function $f$.
\end{observation}

\section{Separating static and dynamic}
\label{sec:sd}

\subsection{The task of an algorithm}
\label{sub:task}

In logic, traditionally, algorithms compute functions. But in the real world, algorithms perform many other tasks and may be intentionally non-terminating. Here, for example and future reference, is a very simple task where a variable signals some undesirable condition.

\begin{task}\label{task}
Keep watching a Boolean variable $b$. Whenever it becomes true, issue an error message, set $b$ to $\false$, and resume watching it. \textqed
\end{task}

To simplify the exposition, we impose the following proviso.

\begin{proviso}
By default, in the rest of this article, the task of an algorithm is to compute a function.
\end{proviso}

The function computed by an algorithm will be called its \emph{objective function}.

\begin{convention}[Input and output]
If an algorithm $A$ computes an $r$-ary objective function $F$, then it has $r$ \emph{input variables} and one \emph{output variable}. These are elementary variables designated to hold the input and output values of $F$. All of the initial states of $A$ are isomorphic except for the values of the input variables. The initial value of the output variable is the default value.
 \textqed
\end{convention}

\subsection{Making dynamic functions initially uninformative}
\label{sub:uninf}

Every dynamic function $f$ has a default value. The generic default value is $\nil$ but, if $f$ is relational, then the default value of $f$ is $\false$.

\begin{definition}
A dynamic function $f$ of an algorithm $A$ is \emph{uninformative} in a state $X$ of $A$ if all its values in $X$ are the default value of $f$.
Function $f$ is \emph{initially uninformative} (for algorithm $A$) if it is uninformative in every initial state of $A$.
\textqed
\end{definition}

The definition of ASMs allows a dynamic function to differ from the default at infinitely many arguments in an initial state and even to be partial there.
One might reasonably stipulate that, initially, every dynamic function $f$ is (total and) uninformative, unless it is an input variable.
Rather than stipulating, however, we can arrive at this desirable situation by a simple transformation of any given algorithm. If the initial configuration of $f$ is preserved as a static function $s$, then the changes made to $f$ can be tracked by a dynamic function $d$ with the help of a dynamic relation $\delta$ indicating the arguments where $f$ has been updated.

\begin{lemma}\label{lem:sd}
Let $A$ be an algorithm of vocabulary $\U\cup\{f\}$ where $f$ is dynamic, not in $\U$, and different from the input and output variables. Let $s,d,\delta$ be fresh function symbols of the arity of $f$ where $s$ is static, $d$ and $\delta$ are dynamic, and $\delta$ is relational. There is an algorithm $B$ of vocabulary  $\U\cup\{s,d,\delta\}$ satisfying the following requirements.
\begin{enumerate}
\item The initial states of $B$ are obtained from those of $A$ by renaming $f$ to $s$ and introducing uninformative $d, \delta$.
\item Every computation $Y_0,Y_1,\dots, Y_j$ of $B$ is obtained from a computation $X_0,X_1,\dots, X_j$ of $A$ in such a way that the following claims hold where for brevity $X=X_j$ and $Y=Y_j$.
  \begin{enumerate}[leftmargin=10pt]
  \item
  \emph{\ITE}$\big(\delta_Y(\x),d_Y(\x),s_Y(\x)\big)=f_X(\x)$.
  \item Every $g\in\U$ has the same interpretation and is evaluated at exactly the same arguments in $X$ and in $Y$.
  \item For every term $t$ in the program of $A$, we have\ \
  $\V_X(t) = \V_Y(\tilde t)$\ \
   where $\tilde t$ is the result of replacing\footnotemark the subterms $f(t')$ of $t$ with terms \emph{\ITE}$\big(\delta_Y(t'),d_Y(t'),s_Y(t')\big)$.
  \item The arguments where $s_Y$ is evaluated are exactly the arguments where $f_X$ is evaluated and where $f$ has not been updated yet.
  \item $B$ fails at $Y$ if and only if $A$ fails at $X$.
  \end{enumerate}
\item $B$ computes the objective function of $A$. \textqed
\end{enumerate}
\end{lemma}

\footnotetext{We have not specified the order in which replacements are done. It is more efficient to use a bottom-up strategy, so that if $f(t_1)$, $f(t_2)$ are subterms of $t$ and $f(t_1)$ is a subterm of $t_2$ then deal with $f(t_1)$ before $f(t_2)$. But the result does not depend on the order of replacements.}

\begin{proof}
To simplify notation, we assume that $f$ is unary and we write $f(x)$ rather than $f(\x)$. The generalization to the case of arbitrary arity will be obvious.
Let $\Pi$ be the program of $A$.
The desired algorithm $B$ is obtained from $A$ in a simple and effective way.
The initial states of $B$ are defined by requirement~1, and the program $\Sigma$ of $B$ is obtained from $\Pi$ in two stages. Recall how the terms $t$ of $\Pi$ are transformed into terms $\tilde t$ in claim~$2(c)$ of the lemma.
\begin{description}
\item[Stage 1] For every term $f(t)$ in $\Pi$, substitute $\ITE(\delta(t),d(t),s(t))$ for every occurrence of $f(t)$ where $f(t)$ isn't the left side of an assignment.
    Let $\tilde\Pi$ be the resulting program.

\item[Stage 2] Replace every assignment $f(\tilde t) := \tilde \tau$ in $\tilde\Pi$ with parallel rule\\ $d(\tilde t) := \tilde \tau\pr \delta(\tilde t) := \true$.
\end{description}
It remains to prove that $B$ works as intended. Since requirement~1 holds by construction and requirement~3 follows from requirement~2, it suffices to prove requirement~2.

\smallskip\noindent
By induction on $j$, we prove claims $(a)$--$(c)$. Assume that the claims have been proved for all $i<j$. Notice that the induction hypothesis and $(a)$ imply $(b)$ and $(c)$. To prove $(a)$, we consider two cases.

Case 1: $\delta_Y(x)=\true$.
Then there is a positive integer $i<j$ such that, at step $i+1$, $\Sigma$ executes a rule $(d(\tilde t):= \tilde \tau \pr \delta(\tilde t):=\true)$ for some $\tilde t$ and $\tilde\tau$ with $\V_{Y_i}(\tilde t) = x$.  There may be several triples $(i, \tilde t, \tilde t')$ fitting the bill; in such a case, fix such a triple with $i$ as big as possible.

By the construction of $\Sigma$, the rule $(d(\tilde t):= \tilde \tau \pr \delta(\tilde t):=\true)$ replaces an assignment $f(\tilde t):= \tilde \tau$ in $\tilde\Pi$ which, in its turn, replaces an assignment $f(t):=\tau$ in $\Pi$. By the induction hypothesis, $\V_{X_i}(t) = \V_{Y_i}(\tilde t) = x$ and $\V_{X_i}(\tau) = \V_{Y_i}(\tilde \tau)$. By the choice of $i$, we have
\begin{align*}
& \ITE(\delta_Y(x),d_Y(x),s_Y(x)) = d_Y(x) = d_{Y_{i+1}}(x) =
  \V_{Y_{i+1}}(d(\tilde t)) =\\
& \V_{Y_i}(\tilde \tau) = \V_{X_i}(\tau) = \V_{X_{i+1}}(f(t)) =
   f_{X_{i+1}}(x) = f_X(x).
\end{align*}

Case 2:  $\delta_Y(x)=\false$. It suffices to prove that program $\Pi$ does not update $f_X$ at $x$ during the $j$-step computation, because then we have
\[\ITE(\delta_Y(x),d_Y(x),s_Y(x)) = s_Y(x) = s_{Y_0}(x) =
 f_{X_0}(x) = f_X(x).\]
Suppose toward a contradiction that an assignment subprogram $f(t):=\tau$ of $\Pi$ is executed at step $i\le j$. But then a subprogram $(d(\tilde t):= \tilde \tau \pr \delta(\tilde t):=\true)$ of $\Sigma$ is executed at stage $i$. By the induction hypothesis, $\V_{Y_i}(\tilde t) = \V_{X_i}(t) = x$ and therefore $\true = \delta_{Y_{i+1}}(x) = \delta_Y(x)$ which contradicts the case hypothesis.

This concludes the proof of claims~$(a)$--$(c)$. They imply the following auxiliary claim.
\begin{enumerate}
\item[$(c')$] $f$ has been updated at $x$ if and only if $\delta_Y(x)=\true$.
\end{enumerate}

\smallskip\noindent
Claim~$(d)$. Taking claims~$(a)$--$(c')$ into account, we have
\begin{align*}
f(x)\ & \textrm{ is evaluated in $X$ and $f(x)$ has not been updated}\\
\iff  & \textrm{$A$ evaluates $f_X(t)$ for some $t$ in $\Pi$ with $\V_X(t)=x$}\\
      & \textrm{and $f(x)$ has not been updated}\\
\iff  & \textrm{$B$ evaluates $f_X(\tilde t)$ for some $t$ in $\Pi$ with $\V_X(t)=x$}\\
      & \textrm{and $\delta(x)=\false$}\\
\iff  & \textrm{$s(x)$ is evaluated in $Y$}
\end{align*}

\smallskip\noindent
Claim~$(e)$. There are two failure scenarios, and we consider them in turn.

Scenario~1: Evaluating a basic function where it is undefined. For some term $t$ in $\Pi$, we have
\begin{align*}
A \textrm{ fai}&\textrm{ls at }X \\
\iff & \Pi\textrm{ attempts to evaluate undefined }\V_X(t) \\
\iff &\Sigma\textrm{ attempts to evaluate undefined }\V_Y(\tilde t) \\
\iff & B\textrm{ fails at }Y.
\end{align*}

Scenario~2: Producing contradicting updates for some basic function.

Case~1: The basic function in question is an $\U$ function $g$. To simplify notation we assume that $g$ is unary. For some terms $t_1,t_2,\tau_1,\tau_2$ in $\Pi$, we have
\begin{align*}
A \textrm{ fai}&\textrm{ls at }X \\
\iff &\Pi\textrm{ attempts to execute } g(t_1):=\tau_1\textrm{ and }g(t_2):=\tau_2\textrm{ in X},\\
&\textrm{where }\V_X(t_1) = \V_X(t_2)\textrm{ but }\V_X(\tau_1) \neq \V_X(\tau_2), \\
\iff &\Sigma\textrm{ attempts to execute } g(\tilde t_1):= \tilde \tau_1\textrm{ and }g(\tilde t_2):= \tilde \tau_2\textrm{ in X},\\
&\textrm{where }\V_Y(\tilde t_1) = \V_X(t_1) = \V_X(t_2) = \V_Y(\tilde t_2) \\
&\textrm{but }\V_Y(\tilde\tau_1) = V_X(\tau_1)\neq \V_X(\tau_2)= \V_Y(\tilde\tau_2), \\
\iff & B\textrm{ fails at }Y
\end{align*}

Case~2: The basic function in question is not an $\U$ function. Then, it must be $f$ in the case of $A$, and it must be $d$ in the case of $\Sigma$. While $\delta$ is also dynamic, the only value assigned to $\delta$ is $\true$. The rest is as in case~1 except that, in the equivalence chain, $g$ is replaced with $f$ in the second item and with $d$ in the third.
\end{proof}

In the lemma, we exempted the output variable from playing the role of $f$ because the output variable is already initially uninformative. There could be other dynamic functions required to be initially uninformatice. For example, there may be nullary dynamic relation \texttt{halt} which is initially $\false$ and which is set to $\true$ when the computation halts. In \S\ref{sec:prune}, we will introduce numerical dynamic functions initially set to their default value $0$.

\begin{theorem}\label{thm:sd}
Let $f_1,f_2,\dots$ be all of the non-input dynamic functions of an algorithm $A$ which are not required \emph{a priori} to be initially uninformative. There is an algorithm $B$, satisfying the following requirements.
\begin{enumerate}
\item The vocabulary of $B$ is obtained from that of $A$ by replacing dynamic symbols $f_m$ with fresh static symbols $s_m$ of the same arity and by adding some dynamic symbols.
\item The initial states of $B$ are obtained from those of $A$ by renaming every $f_m$ to $s_m$ and making all the new  dynamic functions uninformative, so that all non-input dynamic functions of $A$ are initially uninformative.
\item Every computation $Y_0,Y_1,\dots$ of $B$ is obtained from a computation $X_0,X_1,\dots$ of $A$ in such a way that
\begin{enumerate}
  \item $s_m$ is evaluated at $\x$ in $Y_j$ if and only if $f_m$ is evaluated at $\x$ in $X_j$ and $f_m$ has not been updated at $\x$,
  \item for any static function $g$ of $A$, $g_{Y_j} = g_{X_j}$ and $g$ is evaluated at exactly the same arguments in $Y_j$ and in $X_j$, and
  \item $B$ fails at $Y_j$ if and only if $A$ fails at $X_j$.
\end{enumerate}
\item $B$ computes the objective function of $A$.
\end{enumerate}
\end{theorem}

\begin{proof}
Use Lemma~\ref{lem:sd} to replace every $f_m$ with $s_m$ and dynamic bookkeeping functions $d_m$ and $\delta_m$.
\end{proof}

The theorem allows us to impose the following proviso without loss of generality.

\begin{proviso}\label{prv:dynamic}
Below, by default, the non-input dynamic functions of any algorithm are initially uninformative. \textqed
\end{proviso}

\begin{remark}\mbox{}
One generalization of the theorem is related to the task performed by the given algorithm $A$. It does not have to be computing a function. It could be any other reasonable task, e.g.\ Task~\ref{task}. \textqed
\end{remark}

\section{Effectivity: Intuition and tool dependence}
\label{sec:intuit}

Effective algorithms are also known by names like effective procedures and mechanical methods.

\subsection{Intuition}
The notion of effective algorithm has been informal and intuitive. Here is an explanation of it from an influential book on recursive functions and effective computability:

\begin{quoting}
``Several features of the informal notion of algorithm appear to be essential. We describe them in approximate and intuitive terms.
\begin{enumerate}
\item[*1.] An [effective] algorithm is given as a set of instructions of finite size. (Any classical mathematical algorithm, for example, can be described in a finite number of English words.)
\item[*2.] There is a computing agent, usually human, which can react to the instructions and carry out the computations.
\item[*3.] There are facilities for making, storing, and retrieving steps in a computation.
\item[*4.] Let P be a set of instructions as in *1 and L be a computing agent as in *2. Then L reacts to P in such a way that, for any given input, the computation is carried out in a discrete stepwise fashion, without use of continuous methods or analogue devices.
\item[*5.] L reacts to P in such a way that a computation is carried forward deterministically, without resort to random methods or devices, e.g., dice.
\end{enumerate}
Virtually all mathematicians would agree that features *1 to *5, although inexactly stated, are inherent in the idea of algorithm" (Hartley Rogers \cite[\S1.1]{Rogers}). \textqed
\end{quoting}

A numerical function is effective if there is an effective algorithm that computes the function. Everybody agrees that partial recursive functions on natural numbers are effectively computable.

\subsection{Relevance}

There is an important property of effective algorithms that Rogers did not emphasize: There are no restrictions on resources. The computing agent does not run out of time, out of paper, etc. This makes effectivity incompatible with engineering or physics. How is it relevant today?

It often happens in mathematics that it is easier to prove a stronger statement, especially if the stronger statement is cleaner. It is indeed often easier to prove a computational problem ineffective (if we accept the Church-Turing thesis) than to prove that it is not solvable given such and such resources.

There is also, for what it's worth, historical interest in effectivity. But there is something else. Notice that the notion of effectivity readily generalizes to effectivity relative to a given oracle or oracles. There is a good reason for that. Let's have a closer look at item~*2 in the quote above. There is something implicit there which we want to make explicit. The computing agent should be able to ``react to the instructions."
But that ability of the agent depends on the available tools, doesn't it? Working with ruler and compass is different from working with pen and paper. A personal computer may make you more productive than pen and paper.

This leads us to the notion of effectivity relative to the available tools. This more general notion seems to us more interesting and more relevant today.

\section{Means-fit effectivity}
\label{sec:f}

Any static function of an algorithm is essentially an oracle as far as the algorithm is concerned.
Upon invoking/querying a static function on some input, the algorithm waits until, if ever, the oracle provides a reply. If the oracle does not reply then the algorithm is stuck forever.
But not all static functions are equally oracular.

\subsection{Intrinsic and extrinsic}
\label{sub:ie}

In the real world, some static functions are provided by the underlying machinery of the algorithm.
They are part of the normal functionality of the computer system of the algorithm; in that sense they are built-in. These functions are, or at least are supposed to be, provided in a reliable and prompt way. When the algorithm invokes such a built-in function, it gets a value. It may be an error message, if for example the algorithm attempts to divide by 0, but it is a value nevertheless.

Of course, in the real world, things may get more complicated. Much of the functionality of your computer system may be provided via the Internet, and you may lose connection to the Internet. Even if you work offline, your computer system may malfunction.

Here we abstract away from such engineering concerns but we retain the important distinction between the built-in static functions, which we call \emph{intrinsic}, and the other static functions which we call \emph{extrinsic}. To this end, we stipulate that the vocabulary of an algorithm indicates which static symbols are intrinsic. As far as the underlying machinery is concerned, the important part for our purposes is what functions are provided rather than how they are computed. Accordingly, we abstract the underlying machinery of an algorithm to the (in general compound) \emph{datastructure} of the algorithm; see \S\ref{sub:ds} for details.

As far as a given algorithm is concerned, extrinsic functions are provided by the unknown world which is not guaranteed to be prompt or reliable. They are genuine oracles. (In the real world, the extrinsic functions may be also supported by reliable computer systems; we address this in \S\ref{sec:prune}).)
Let us see some examples.

\begin{example}[Ruler and compass]\label{ex:rc}
In the historically important realm of ruler-and-compass algorithms, the underlying machinery includes (unmarked) ruler and compass. These algorithms are not effective in the Church-Turing sense because of the analog, continuous nature of their basic operations, but they had been practical in antiquity. (They also admit limited nondeterminism but they can be made deterministic by means of some simple choice functions.) \textqed
\end{example}

\begin{example}[Idealized human]
An algorithm can be executed by a human being. Viewing an (idealized) human as the underlying machinery of an algorithm is a key idea in Turing's celebrated analysis \cite{Turing}. \textqed
\end{example}

The datastructure of natural numbers $0,1,2,\dots,$ with the standard operations will be called \emph{natural-numbers arithmetic} or simply \emph{arithmetic}. The exact set of standard operations does not matter as long as we have (directly or via programming) $0$ and the successor operation. We will assume here that the standard operations are zero, successor operation and predecessor operation.

\begin{example}[Recursion theory]\label{exl:rec}
Natural-numbers arithmetic is the datastructure of traditional recursion theory, studying partial recursive functions \cite{Rogers}. \textqed
\end{example}

\begin{example}[Turing machines]\label{exl:tm}
Consider Turing machines with a single tape which is one-way infinite and has only finitely many non-blank symbols.
The datastructure of such machines is composed of three finite datatypes --- control states, tape symbols and movement directions --- and one infinite datatype of tape cells. All by itself, the infinite datatype is isomorphic to the natural-numbers arithmetic. \textqed
\end{example}

\begin{example}[Programming language]
A modern programming language may involve a number of datastructures. We see them all as parts of one compound datastructure which is the datastructure of any algorithm written in the programming language. The algorithm may also query some additional functions online. These outside functions would be extrinsic. \textqed
\end{example}

\begin{example}[Random access machines]\label{exl:ram}
Real-world computers are too messy for many theoretical purposes.
This led to an abstract computation model called Random Access Machine, in short RAM \cite{CR}, whose datastructure is more involved than arithmetic but can be encoded in arithmetic. \textqed
\end{example}

\begin{example}[BSS machines]\label{exl:bss}
The datastructure of Blum-Shub-Smale machines, also known as BSS machines  \cite{BSS}, involves natural numbers and genuine reals. BSS machines are able to compute functions over the reals which, for the obvious reason, cannot be computed by Turing machines. But the numerical functions computed by BSS machines are partial recursive. \textqed
\end{example}

\subsection{Datastructures}
\label{sub:ds}

In logic terms, a datastructure is a many-sorted first-order structure. Since our algorithms are abstract state machines, the datastructure of an algorithm always includes the logic sort comprising elements $\true$, $\false$, and $\nil$.
Another sort could be that of natural numbers. In the case of BSS machines, we have also the sort of genuine reals.
In the case of ruler-and-compass algorithms, we have three nonlogic sorts: points, straight lines, and circles of the real plane \cite{Beeson}.

For our purposes, it will be convenient to see datastructures as ordinary structures, with just one base set, where the sorts are represented by static unary relations, a.k.a.\ characteristic functions. If $A$ is an algorithm with datastructure $D$ then any initial state $X$ of $A$ is an extension of $D$ with (i)~dynamic functions which, with the exception of input variables, are uninformative, and (ii)~extrinsic functions, if any.

\subsection{Definition}
\label{sub:fdef}

In a somewhat anticlimactic way, the intrinsic/extrinsic dichotomy allows us to characterize algorithms which are effective relative to their datastructures and are means bound in that sense. We formulate this characterization as a definition.

\begin{definition}[Means-fit effectivity]
\label{def:F}
An algorithm $A$ is \emph{means-fit effective} or \emph{effective relative to its datastructure} if $A$ has no extrinsic functions.
\end{definition}

The relative-to-datastructure character of effectivity is natural. The underlying datastructure matters. The effectivity of arithmetic-based algorithms is different from the  effectivity of ruler-and-compass algorithms and from the effectivity of BSS algorithms.

Definition~\ref{def:F} is especially natural if the extrinsic functions of the algorithm in question are genuine oracles.
But what if some of those oracles are computable, that is, computable by algorithms effective relative to their respective datastructures?
In the real world, new software is rarely written from scratch. Typically, it reuses pieces of software which reuse other pieces of software, and so on. And this is not a strict hierarchy in general; the dependencies between various pieces may be more complicated.

Can we expand a given algorithm $A$ so that it incorporates the auxiliary algorithms behind $A$'s computable extrinsic functions and, for each of these auxiliary algorithms $B$, the algorithms behind $B$'s computable extrinsic functions, etc.?
It turns out that the answer is positive if only finitely many algorithms are involved altogether. The following two sections are devoted to proving this result. After that, we will return to the discussion of effectivity.

\section{Query serialization}
\label{sec:ser}

An algorithm may produce, during one step, many extrinsic queries, that is queries to extrinsic functions. This is problematic for the purpose of the next section. Fortunately query production can be serialized as the following Theorem~\ref{thm:ser} shows. The theorem is of independent interest but, in this paper, it is just an auxiliary result to be used in the following section.

We say that an ASM program $\Pi$ is a \emph{compound conditional composition} or a \emph{compound conditional} if it has the form
\begin{align*}
\If g_1\ &\ \Then P_1\\
\Elseif g_2\ &\ \Then P_2\\
\vdots\\
\Elseif g_n\ &\ \Then P_n
\end{align*}
If each $P_i$ is a parallel composition of assignments, then $\Pi$ is a compound conditional of parallel assignments.

\begin{lemma}\label{lem:ser}
For every ASM $A$ there is a behaviourally equivalent ASM $A'$ such that the program of $A'$ is a compound conditional of parallel assignments and for every state $X$ of $A$ (and thus state of $A'$ as well), $A$ and $A'$ generate exactly the same extrinsic queries at $X$. \textqed
\end{lemma}

\begin{proof}
It suffices to prove the following claim. Let $\Pi = (P \pr Q)$ where $P,Q$ are compound conditionals of parallel assignments.  There is a compound conditional of parallel assignments $\Pi'$ of the vocabulary of $\Pi$ such that for every state $X$ of $\Pi$ (and of $\Pi'$), we have
\begin{enumerate}
\item[(i)] $\Pi$ and $\Pi'$ generate the same updates at $X$, and
\item[(ii)] $\Pi$ and $\Pi'$ generate the same extrinsic queries at $X$.
\end{enumerate}
We illustrate the proof of the claim on an example where $\Pi$ is

\noindent
\begin{minipage}{.45\textwidth}
\begin{align*}
\If g_1\ &\ \Then P_1\\
\qquad\Elseif g_2\ &\ \Then P_2
\end{align*}
\end{minipage}
\begin{minipage}{.05\textwidth}
 $\prl$
\end{minipage}
\begin{minipage}{.45\textwidth}
\begin{align*}
\If h_1\ &\ \Then Q_1\phantom{mmmmm}
\end{align*}
\end{minipage}

\bigskip\noindent
respectively. The desired $\Pi'$ is
\begin{align*}
\If g_1\land h_1\quad &\ \Then P_1\prl Q_1\\
\Elseif g_1\quad &\ \Then P_1\\
\Elseif g_2\land h_1\quad &\ \Then P_2\prl Q_1\\
\Elseif g_2\quad &\ \Then P_2\\
\Elseif h_1\quad &\ \Then Q_1\\
\end{align*}
All states $X$ of $\Pi$ and $\Pi'$ split into six categories depending on which, if any, of the 5 guards holds in $X$. It is easy to check, for each of the six categories, that $\Pi$ and $\Pi'$ generate the same extrinsic queries. Consider for example, a state $X$ satisfying $(\neg g_1\land g_2)\land h_1$. Both $\Pi$ and $\Pi'$ evaluate the same guards $g_1, g_2$ and $h_1$ and execute the same rule $P_2\pr Q_1$  at $X$, and therefore the requirements (i) and (ii) are satisfied.
\end{proof}

\begin{definition}
An algorithm $A'$ \emph{tightly elaborates} an algorithm $A$ if the vocabulary of $A'$ is that of $A$ plus some auxiliary elementary variables and if the following two conditions are satisfied where the \emph{default expansion} $X'$ of a state $X$ of $A$ is obtained by setting the auxiliary variables to their default values.
\begin{enumerate}
\item Every single step $X,Y$ of $A$ (where $Y$ is the result of executing $A$ at $X$) gives rise to a unique computation $X',\dots, Y'$ of $A'$ called a \emph{mega-step}. The \emph{mega-step} $X',\dots, Y'$ is composed from a bounded number of regular steps of $A'$. The extrinsic queries issued by $A'$ during the mega-step $X',\dots, Y'$ are exactly the extrinsic queries issued by $A$ during the step $X,Y$.
\item Every finite (resp. infinite) computation of $A'$ has the form
    \[\ X'_0,\dots,X'_1,\dots,X'_2,\dots,X'_N
    \quad\big(\textrm{resp. }\
    X'_0,\dots,X'_1,\dots,X'_2,\dots\big) \]
    \[\ \textrm{where}\quad X_0,\ X_1,\ X_2,\dots,X_N
    \quad\big(\textrm{resp. }\
    X_0,\ X_1,\ X_2,\ X_3,\dots\big) \]
    is a finite (resp. infinite) computation of $A$. \textqed
\end{enumerate}
\end{definition}

Notice that $A$ is function-computing if and only if $A'$ is, and then they compute the same objective function.

\begin{theorem}\label{thm:ser}
For any algorithm $A$, there is an algorithm $A'$ which tightly elaborates $A$ and produces at most one extrinsic query per regular step. \textqed
\end{theorem}

\begin{proof}
In virtue of Lemma~\ref{lem:ser}, we may assume without loss of generality that the program $\Pi$ of $A$ is a compound conditional of parallel assignments.

The program of the desired algorithm $A'$ has the form
\[ \If \neg \Done\ \Then \Pi'\ \Else \Done:=\false \]
where $\Pi'$ is the \emph{meaningful part} of the program and $\Done$ is an auxiliary Boolean variable.
We construct $\Pi'$ by induction on $\Pi$.

\smallskip\noindent
\texttt{Basis of induction:} $\Pi$ is a parallel composition of assignments.
Form a list
\[ t_1,\ t_2,\ t_3,\ \dots,\ t_{n-1},\ t_n\]
of all extrinsic-head terms in $\Pi$ such that if $t_j$ is a subterm of $t_k$ then $j<k$.
Let $d_1, d_2, \dots, d_n$ be fresh elementary variables.
The plan is to evaluate $t_1,\dots,t_n$ one by one and store the results in $d_1,\dots,d_n$ respectively, and then to complete the job of the original assignment without any extrinsic calls.
This plan requires us to be careful with substitutions.
Construct a matrix
\begin{align*}
& t_1^0,\ t_2^0,\ t_3^0,\ \dots,\ t_{n-1}^0,\ t_n^0 \\
& d_1,\ t_2^1,\ t_3^1,\ \dots,\ t_{n-1}^1,\ t_n^1 \\
& d_1,\ d_2,\ t_3^2,\ \dots,\ t_{n-1}^2,\ t_n^2 \\
& \vdots\\
& d_1,\ d_2,\ d_3,\ \dots,\ d_{n-1},\ t_n^{n-1}\\
& d_1,\ d_2,\ d_3,\ \dots,\ d_{n-1},\ d_n
\end{align*}
where $t_j^0 = t_j$ and if $i>0$ then $t_j^i = t_j^{i-1}\{d_i\mapsto t_i^{i-1}\}$, so that every instance of $t_i^{i-1}$ in $t_j^i$ is replaced with $d_i$. Similarly, construct programs
\[ \Pi^0,\ \Pi^1,\ \Pi^2,\ \Pi^3,\ \dots,\ \Pi^{n-1},\ \Pi^n\]
where $\Pi^0=\Pi$, and if $i>0$ then $\Pi^i = \Pi^{i-1}\{d_i\mapsto t_i^{i-1}\}$. Notice that $\Pi^n$ has no extrinsic functions.

Let $b_1, b_2, \dots, b_n$ be fresh Boolean variables. Recall that every $b_i$ is initially false and therefore every $\neg b_i$ is initially true. The desired program $\Pi'$ is
\begin{align*}
\If\neg b_1\
&\ \Then\big(d_1 := t_1^0 \prl b_1:=\true\big) \\
\Elseif\neg b_2\
&\ \Then\big(d_2:=t_2^1\prl b_2:=\true\big)\\
\vdots \\
\Elseif\neg b_n\
&\ \Then
\big(d_n :=t_n^{n-1} \prl b_n:=\true\big) \\
\Else\quad &\quad (\Pi^n \prl \Done:=\true\prl\\
           &\quad\ b_1:=\false\prl\dots\prl b_n:=\false)
\end{align*}
This completes the basis of our induction.

\smallskip\noindent
\texttt{Induction step:} $\Pi$ has the form\qquad
$\If \beta\ \Then P\ \Else Q$\\
where $P,Q$ are compound conditionals of parallel assignments.
By the induction hypothesis, there are programs
\begin{align*}
 \If \neg \Done^P\ & \Then P'\ \Else \Done^P:=\false, \\
 \If \neg \Done^Q\ & \Then Q'\ \Else \Done^Q:=\false
\end{align*}
which tightly elaborate $P, Q$ respectively and produce at most one extrinsic query per regular step.

The desired algorithm $A'$ starts with evaluating $\beta$. This part is exactly like the induction basis except that we are given a term $\beta$ rather than an assignment. In particular, now
$t_1, t_2, \dots, t_n$ are all of the extrinsic-head terms in $\beta$, and instead of programs $\Pi^0,\ \Pi^1,\ \Pi^2,\ \Pi^3,\ \dots,\ \Pi^{n-1},\ \Pi^n$ we have terms
\[ \beta^0,\ \beta^1,\ \beta^2,\ \beta^3,\ \dots,\ \beta^{n-1},\ \beta^n\]
where $\beta^0=\beta$ and if $i>0$ then $\beta^i = \beta^{i-1}\{d_i\mapsto t_i^{i-1}\}$.
The desired program $\Pi'$ uses additional auxiliary Boolean variables $a$ and $b$.

\begin{align*}
\If&\ \neg a\land\neg b_1\
\ \Then (d_1 := t_1^0 \prl b_1:=\true) \\
\Elseif&\ \neg a\land\neg b_2\
\ \Then (d_2 := t_2^1 \prl b_2:=\true) \\
&\vdots \\
\Elseif&\ \neg a\land\neg b_n\
\ \Then (d_n := t_n^{n-1} \prl b_n:=\true)\\
\Elseif&\ \neg a\
\ \Then (b:=\beta^n \prl a:=\true\prl \\
&b_1:=\false\prl\dots\prl b_n:=\false)\\
\Elseif&\ a\land b\land(\neg\Done^P)\ \ \Then P' \\
\Elseif&\ a\land\neg b\land(\neg\Done^Q)\ \ \Then Q'\\
\Else&\ (\Done:=\true \prl a:=\false \prl \\
       &\Done^P:=\false \prl \Done^Q:=\false)
\phantom{m}\qedhere
\end{align*}
\end{proof}

\begin{remark}
For those interested in details, let us illustrate two subtleties which complicate the proof of the theorem.
First, consider the program $\Pi'$ of the induction step.
One may be tempted to eliminate the else clause and replace $P',Q'$ with $P'\pr\Done:=\true$ and $Q'\pr\Done:=\true$ respectively. But this does not work because $\Done$ will be set to $\true$ on the very first regular step of $P', Q'$ and thus will prevent them from completing their mega-step.

Second, why did we bother with proving that lemma? It seems that we could consider another case of the induction step where $\Pi = P\pr Q$. The desired $\Pi'$ could be

\vspace{-3em}
\begin{align*}
\If\ \neg\Done^P\ &\ \Then P' \\
\Elseif\ \neg\Done^Q\ &\ \Then Q'\\
\Else\ &\ \Done:=\true
\end{align*}
But it is risky to serialize parallel composition. Just consider the case where $P,Q$ are $a:=b$ and $b:=a$ respectively. \textqed
\end{remark}

For future use we record the following corollary.

\begin{corollary}\label{cor:ser}\rm
\textit{The program of $A'$ has the form}
\[ \If \neg \Done\ \Then \Pi'\ \Else \Done:=\false \]
\textit{where $\Pi'$ has the form}
\begin{align*}
\If g_1 &\ \Then R_1 \\
\Elseif g_2 &\ \Then R_2 \\
&\vdots \\
\Elseif g_n &\ \Then R_n
\end{align*}
\textit{and it is last rule $R_n$ that sets \emph{\Done} to $\true$; the rules $R_1,\dots,R_{n-1}$ do not mention \emph{\Done}.
Furthermore, the $n$ clauses split into two categories:}
\begin{description}
\item[\texttt{Pure}] \textit{clauses with no occurrences of extrinsic functions, and}
\item[\texttt{Tainted}] \textit{clauses of the form}\quad
\texttt{[else]if $g_k$ then $d_k:=t_k\pr Q_k$} \\
\textit{where $d_k$ is an elementary variable, the head function of $t_k$ is extrinsic and there are no other occurrences of extrinsic functions.} \textqed
\end{description}
\end{corollary}


\section{Pruning off effective oracles}
\label{sec:prune}

\begin{definition}[Numeric algorithms]\mbox{}
\begin{enumerate}
\item An algorithm (that is an ASM) $A$ is \emph{numeric} if, up to isomorphism, every state of $A$ incorporates natural-numbers arithmetic.
\item Some function symbols in the vocabulary of a numeric algorithm $A$ are marked \emph{numerical}. A basic numerical function takes only nonnegative integer values, and its default value is $0$. \textqed
\end{enumerate}
\end{definition}

In \S\ref{sec:asm}, we defined the notions of consistency and union of structures.
Since we view datastructures as first-order structures (see \S\ref{sec:f}), we have the notions of consistency and union of datastructures.

\begin{theorem}[Pruning Theorem]\label{thm:prune}
Consider $N$ algorithms
$A_0, A_1, \dots, A_{N-1}$
whose datastructures are pairwise consistent,
and suppose that, for every $A_i$, every extrinsic function of $A_i$ is the objective function of some $A_j$. There is a numerical algorithm $B$, with no extrinsic functions, such that
\begin{enumerate}
\item the datastructure of $B$ is the union of the $N$ datastructures of the algorithms $A_i$ plus natural-numbers arithmetic.
\item $B$ computes the objective function of $A_0$. \textqed
\end{enumerate}
\end{theorem}

\begin{remark}\mbox{}
\begin{enumerate}
\item If the algorithms $A_i$ are numerical and have the same datastructure $D$, then $D$ is also the datastructure of $B$.
\item The algorithm $B$, constructed in the proof, works exactly like $A_0$ except that, instead of waiting for extrinsic queries to be answered, $B$ computes the answers. It is therefore no wonder that $B$ executes the task of $A_0$ which happens to be computing the objective function. But it could be virtually any other reasonable task, e.g.\ Task~\ref{task}. \textqed
\end{enumerate}
\end{remark}

\begin{proof}[Proof of theorem]
Without loss of generality, different algorithms $A_i$ compute different objective functions. Indeed, if $m<n$ and $A_m, A_n$ compute the same objective function, remove $A_n$. The remaining set is still closed in the sense that, for every $A_i$, every extrinsic function of $A_i$ is computed by some $A_j$.

In essence, we have a mutually recursive system of algorithms.
Whenever one algorithm $A_i$ poses a query to an extrinsic function powered by algorithm $A_j$ (which may be $A_i$ itself), it implicitly calls $A_j$. The caller suspends its execution and transfers the execution control to the callee.
If and when the execution control returns from the callee, the caller resumes executing from the point where it put itself on hold, and it needs the exact same data it was working on except that the query is answered.

Following the standard practice, we use a call stack to implement recursion. To this end, we introduce a numerical variable $\top$ (read ``top," not ``true"), indicating the current position of the top of the stack, and a unary dynamic function $\Active$ such that $\Active(\top)$ indicates the index $i$ of the currently active algorithm.

However we have two non-standard difficulties which make our work harder. One of them is that an algorithm may generate many extrinsic queries in one step. That difficulty was addressed in the previous section.
To illustrate the other difficulty, consider this scenario.  $A_1$ calls $A_2$ and increases $\top$, say from $7$ to $8$. When $A_2$ is done and $\top$ is decreased to $7$, it would be useful to defaultify $A_2$, that is to make all its dynamic functions uninformative.
But defaultification, a form of garbage collection, is a problem in our abstract setting. Fortunately --- and ironically --- effectivity does not have to be efficient. Instead of cleaning up a used copy of $A_2$, we abandon it and use instead another, fresh copy of $A_2$. To this end we need an ample supply of copies of every $A_i$. This will be achieved by means of a special parameter $\n$.

In accordance with Corollary~\ref{cor:ser}, we may assume that the program of algorithm $A_i$ has the form
\[ \If \neg \Done_i\ \Then \Pi_i\ \Else \Done_i:=\false \]
where the meaningful part $\Pi_i$ has the form
\begin{align*}
\If g_{i1} &\ \Then R_{i1} \\
\Elseif g_{i2} &\ \Then R_{i2} \\
&\vdots \\
\Elseif g_{in_i} &\ \Then R_{in_i}\\
\end{align*}
described in Corollary~\ref{cor:ser}. In particular, tainted clauses have the form
\begin{equation}\label{eq:taint}
\texttt{[else]if $g_{ik}$ then $d_{ik}:=t_{ik} \pr Q_{ik}$}
\end{equation}

To construct the program of $B$, we modify programs $\Pi_i$ in three stages.

\smallskip\noindent
\texttt{Stage~1.} $\Pi_i'$ is obtained from $\Pi_i$ by expanding every dynamic function $d(x_1,\dots,x_r)$ to $d(\n,x_1,\dots,x_r)$, so that $d$ acquires a new argument position, ahead of all the old argument positions, and this new position is occupied by a numerical variable $\n$.
The expansion transforms any term $t$ to a term $t'$ and any rule $R$ to a rule $R'$.
In particular, the tainted clauses of $\Pi_i$ acquire this form:
\begin{equation}\label{eq:prime}
\texttt{[else]if $g_{ik}'$ then
$d_{ik}(\n):=t_{ik}' \pr Q_k'$}
\end{equation}
As a result we have infinite supply of copies of $\Pi_i$.

Notice that the same parameter $\n$ is used for all programs $\Pi_i'$. We will have one consecutive numbering of copies of all these programs. We will speak about sessions. Each value of $\n$ corresponds to one session, and a particular $A_i$ executes during that session. Each extrinsic call will create a new session. A special numerical variable $\Max$ will keep track of the maximal session number so far.
For example, suppose that $A_1$ executes session $10$, $\top=7$, $\Max=20$, and $A_1$ calls $A_2$. This will start session $21$, increment $\top$ to $8$ and set $\Max$ to $21$.

The input and output variables of $A_i$ will be denoted $\IN_{i1}$, $\IN_{i2}$, \dots, $\OUT_i$.
The input and output variables of the desired algorithm $B$ will be denoted $\IN_1$, $\IN_2$, \dots, $\OUT$.

\smallskip\noindent
\texttt{Stage~2.} $\Pi_i''$ is obtained from $\Pi_i'$ by modifying every tainted clause \eqref{eq:prime}.
If $t_{ik}' = e(\tau_1',\dots,\tau_r')$ and $e$ is computed by $A_j$, then replace the assignment $d_{ik}(\n):=t_{ik}'$ with the following administrative program which we display first and then explain.
\begin{align*}
&\If\ \neg b_{ik}\ \Then\ \top:=\top+1\prl\Active(\top+1):=j\prl\\
&\qquad\n:=\Max+1\prl\Max:=\Max+1\prl \TOP(\Max+1):=\top\prl\\
&\qquad \IN_{j1}(\Max+1):=\tau_1'\prl \dots \prl\IN_{jr}(\Max+1):=\tau_r' \prl\\
&\qquad \Ret(\Max+1):=\n \prl b_{ik}:=\true\\
&\Elseif\ \TOP(\n)=\top\ \Then\\
&\qquad d_{ik}(\n):=\To(\n) \prl Q_{ik}' \prl b_{ik}:=\false
\end{align*}

Thanks to $b_{ik}$, the administrative program works in two steps. On the first step it increments $\top$ and passes control to $A_j$ which involves a number of details.
A new copy of $A_j$ is being engaged and given number $\Max+1$ where $\Max$ holds the maximal value of $\n$ used so far. Accordingly $\Max$ is incremented. A special dynamic function $\TOP(\n)$ records the value of $\top$ corresponding to the $\n^{th}$ program copy in the consecutive numbering of all copies of all programs $\Pi_i$. Further, the proper input information for $A_j$ is supplied. Finally, using a special unary function $\Ret$ (alluding to ``return address"), $A_j$ is notified that, upon termination, it should return control to the current copy of $A_i$.

The administrative program executes the second step only when (and if) $A_j$ computes the desired output and passes the control and the output to $A_i$. The output is passed by means of a special function $\To$. But how did $A_j$ know who to pass the output to? For this we need to see how $\Pi_j''$ was modified at Stage~3.

\smallskip\noindent
\texttt{Stage~3.} $\Pi_i^+$ is the program
\begin{align*}
&\If\ \neg\Done_i(\n)\ \Then\ \Pi_i''\\
&\Else\ \top:=\top-1 \prl \n:=\Ret(\n)\prl\\
&\qquad \To(\Ret(\n)):=\OUT_i(\n)
\end{align*}

To explain what goes on, let us return to the scenario where $A_1$ executed session $10$, $\top$ was 7, $\Max$ was 20, and $A_1$ passed control to $A_2$ which started session $21$, incremented $\top$ to $8$ and set $\Max$ to $21$. On the same occasion, $A_1$ set $\Ret(21)$ to 10. If and when $A_2$ finishes session 21, it will decrement $\top$ to 7 and will pass control to the program of session 10 which happens to be $A_1$. On the same occasion, $A_2$ will assign its output value to $\To(10)$.

Finally, the program of the desired algorithm $B$ is
\begin{align*}
\If \neg \initialized\ &\quad
  \Then \texttt{Initialize}\\
\Elseif \OUT_1(0)=\nil\ &\quad \Then \texttt{Toil}\\
\Else &\quad \texttt{Finish}
\end{align*}
where \initialized\ is a fresh Boolean variable and the three constituent programs are as follows:

\smallskip\noindent
\texttt{Initialize}\ passes the inputs of $B$ to the $0^{th}$ copy of $A_0$ and sets \initialized\ to $\true$:
\begin{align*}
& \IN_{01}(0):=\IN_1 \prl \dots \prl \IN_{0r}(0):=\IN_r\\
& \textrm{initialized}:=\true
\end{align*}
\noindent
where $r$ is the arity of the objective function of $A_0$.

\smallskip\noindent
\texttt{Toil} is the parallel composition of rules
\[ \If \Active(i)\ \Then \Pi_i^{+} \]
where $i = 0, \dots, {N-1}$.

\smallskip\noindent
\texttt{Finish} just passes the output from the $0^{th}$ copy of $A_0$ to $B$:
\[ \OUT := \OUT_0(0)  \]
This completes the proof of Theorem~\ref{thm:prune}.
\end{proof}

In the obvious way, Theorem~\ref{thm:prune} relativizes in the standard sense of computation theory. More explicitly, the theorem remains true if the algorithms $A_i$ have access to some genuine oracles $O_{i1}, O_{i2},\dots$ and if $B$ has access to all these oracles.

\section{Absolute effectivity, and related work}
\label{discuss}

As we saw above, different algorithms may use different datastructures. But, in logic, one datastructure has been playing the central role historically. This is arithmetic of course. Call an algorithm \emph{arithmetical} if its datastructure is arithmetic.

\begin{definition}
A numerical function $F$ is \emph{effective relative to artihmetic} if it computable by an effective arithmetical algorithm. \textqed
\end{definition}

As we mentioned in the introduction, it is the historical, mathematical Church-Turing thesis that we are discussing here.
The minimum task needed to verify the thesis is to prove the claim that every numerical function effective relative to arithmetic is partial recursive.
In \cite{G188}, this claim has been derived from the axioms of \cite{G141} plus an initial-state axiom asserting that, in the initial states, the dynamic functions --- with the exception of input variables --- are uninformative. The current paper provides additional justification for this axiom.

A more ambitious task would be to show that, for any reasonable datastructure $D$, the numerical functions computed by the $D$ based algorithms are partial recursive. The obvious problem is: Which datastructures are reasonable? This allows us to segue into our quick review of related work.

The idea to define effectivity by expanding the axiomatization of \cite{G141} with an initial-state axiom was enunciated by Udi Boker and Nachum Dershowitz in \cite{BD2006} and elaborated in \cite{BD2008}.
Their Initial Data Axiom asserts that an initial state comprises a Herbrand universe, some finite data, and ``effective oracles" in addition to input.
They shifted focus from single algorithms to computation models.
``The resultant class \dots includes all known Turing-complete state-transition models, operating over any countable domain." Any computation model satisfying the axioms of \cite{G141} plus the Initial Data Axiom is ``of equivalent computational power to, or weaker than, Turing machines."

In \cite{G188}, where the minimum task was accomplished, that task was extended to include any algorithms over a countable domain equipped with an injective mapping from that domain to the natural numbers, subject to some recursivity requirement.

A very different approach was taken by Wolfgang Reisig in \cite{Reisig}. Every first-order structure $S$ of vocabulary $\U$ imposes an equivalence relation on $\U$ terms: $t_1\sim_S t_2\iff\V_S(t_1) = \V_S(t_2)$. ``The successor state $\M(S)$ of a state $S$ is fully specified by the equivalence $\sim_S$\ \dots\ Consequently, $\M(S)$ is computable in case $\sim_S$ is decidable. Furthermore, this result implies a notion of computability for general structures, e.g. for algorithms operating on real numbers."

In \cite{BD2010}, Boker and Dershowitz observe that Reisig's approach gives rise to a novel definition of effectivity, though they restrict attention to the case where every element of any initial state is nameable by a term. They show that the three definitions of equivalence --- in \cite{G188}, in \cite{BD2008}, and implicitly in \cite{Reisig} ---  coincide.

In \cite{DF}, Dershowitz and Falkovich ``describe axiomatizations of several aspects of effectiveness: effectiveness of transitions; effectiveness relative to oracles; and absolute effectiveness, as posited by the Church-Turing Thesis."

This concludes our quick review of related literature.

\end{document}